\documentclass[copyright,creativecommons]{eptcs} 
\providecommand{\event}{GANDALF 2013} 
\usepackage{amsmath,amssymb,stmaryrd,xspace,tikz,ifthen,wasysym}
\usepackage{algorithm, algorithmicx}
\usepackage{algpseudocode}
\usepackage[utf8x]{inputenc}
\usepackage{pgfplots}
\usepackage{pgfplotstable}

\usetikzlibrary{arrows,automata,shapes}

\newtheorem{defi}{Definition}
\newtheorem{theo}{Theorem}
\newtheorem{prop}{Proposition}
\newtheorem{lemm}{Lemma}
\newtheorem{coro}{Corollary}

\newenvironment{definition}{\begin{defi} \rm }{\end{defi}}
\newenvironment{theorem}{\begin{theo} \rm }{\end{theo}}
\newenvironment{proposition}{\begin{prop} \rm }{\end{prop}}
\newenvironment{lemma}{\begin{lemm} \rm }{\end{lemm}}

\newenvironment{proof}{\begin{trivlist} \item[\hspace{\labelsep}\bf Proof:]}{\hfill$\Box$\end{trivlist}}

\def\nat{\mathbb{N}}

\renewcommand{\i}{\ensuremath{\ocircle}\xspace}
\newcommand{\odd}{\ensuremath{\square}\xspace}
\newcommand{\even}{\ensuremath{\Diamond}\xspace}

\newcommand{\myomit}[1]{}

\newcommand{\solitaire}[1][]{\ensuremath{\mc{G}^{#1}}\xspace}
\newcommand{\whitegame}[1][]{\ensuremath{\mc{M}^{#1}}\xspace}
\newcommand{\weak}[1][]{\ensuremath{\mc{W}^{#1}}\xspace}

\newcommand{\ziel}{\textsc{\small Zielonka}\xspace}

\newcommand{\pnot}[1]{\bar{#1}}

\newcommand{\attrsym}{\ensuremath{\textit{Attr}}}
\newcommand{\attr}[3][]{\ensuremath{\attrsym^{#1}_{#2}(#3)}}

\newcommand{\mc}[1]{\ensuremath{\mathcal{#1}}}
\newcommand{\ie}{\emph{i.e.}\xspace}
\newcommand{\eg}{\emph{e.g.}\xspace}

\newcommand{\etal}{\emph{et al.}\xspace}

\newcommand{\oftype}{{:}}

\newcommand{\priosym}{\mathcal{P}}
\newcommand{\prio}[1]{\priosym(#1)}

\begin{document}

\title{Zielonka's Recursive Algorithm:\\
dull, weak and solitaire games and tighter bounds}
\author{Maciej Gazda and Tim A.C. Willemse
\institute{Eindhoven University of Technology, The Netherlands}
}

\def\titlerunning{Zielonka's Recursive Algorithm:
dull, weak and solitaire games and tighter bounds}
\def\authorrunning{Maciej Gazda and Tim A.C. Willemse}

\maketitle

\begin{abstract}
Dull, weak and nested solitaire games are important classes of parity
games, capturing, among others, alternation-free $\mu$-calculus
and ECTL$^*$ model checking problems. These classes can be solved in
polynomial time using dedicated algorithms.  We investigate the complexity
of Zielonka's \emph{Recursive} algorithm for solving these special
games, showing that the algorithm runs in $\mc{O}(d\cdot(n+m))$ on weak
games, and, somewhat surprisingly, that it requires exponential time to
solve dull games and (nested) solitaire games. For the latter classes,
we provide a family of games $\solitaire$, allowing us to establish
a lower bound of $\Omega(2^{n/3})$.
We show that an optimisation of Zielonka's
algorithm permits solving games from all three classes in polynomial
time. Moreover, we show that there is a family of (non-special) games
$\whitegame$ that permits us to establish a lower bound
of $\Omega(2^{n/3})$, improving on the previous lower bound for the
algorithm.

\end{abstract}

\section{Introduction}
\label{sec:introduction}

Parity games \cite{EJ:91,McN:93,Zie:98} are infinite duration, two player
games played on a finite directed graph. Each vertex in the graph is
owned by one of the two players and vertices are assigned a priority.
The game is played by moving a single token along the edges in the
graph; the choice where to move next is decided by the player owning the
vertex on which the token currently resides. A parity winning condition
determines the winner of this infinite play; a vertex in the game is won
by the player that can play such that, no matter how the opponent plays,
every play from that vertex is won by her, and the winner of each vertex is
uniquely determined~\cite{McN:93}.  From a practical point
of view, parity games are interesting as they underpin verification,
satisfiability and synthesis problems, see~\cite{CPPW:07,EJ:91,AVW:03}.

The simplicity of the gameplay is fiendishly deceptive. Despite continued
effort, no polynomial algorithm for solving such games (\ie computing the
set of vertices won by each player) has been found.  Solving a parity game
is known to be in $\text{UP}\cap\text{coUP}$~\cite{Jur:98}, a class that
neither precludes nor predicts the existence of a polynomial algorithm. In
the past, non-trivial classes of parity games have been identified
for which polynomial time solving algorithms exist.  These classes
include \emph{weak} and \emph{dull} games, which arise naturally from
alternation-free modal $\mu$-calculus model checking, see~\cite{BG:04},
and \emph{nested solitaire} games which are obtained from \eg the $L_2$
fragment of the modal $\mu$-calculus, see~\cite{BG:04,EJS:01}. Weak and
dull games can be solved in $\mc{O}(n+m)$, where $n$ is the number
of vertices and $m$ is the number of edges, whereas (nested) solitaire
games can be solved in $\mc{O}(\log(d) \cdot (n+m))$, where $d$ is
the number of different priorities in the game.

One of the most fundamental algorithms for solving parity games
is Zielonka's \emph{recursive algorithm}~\cite{Zie:98}. With a complexity of
$\mc{O}(n^d)$, the algorithm is theoretically less attractive
than \eg Jurdzi\'nski's \emph{small progress measures} algorithm~\cite{Jur:00},
Schewe's \emph{bigstep} algorithm~\cite{Sch:07} or the sub-exponential
algorithm due to Jurdzi\'nski \etal\cite{JPZ:06}. However, as observed
in~\cite{FL:09}, Zielonka's algorithm is particularly effective in practice,
typically beating other algorithms.  In view of this, one might therefore
ask whether the algorithm is particularly apt at solving natural
classes of games, taking advantage of the special structure of these
games. We explore this question by investigating the complexity of 
solving weak, dull and
nested solitaire classes using Zielonka's algorithm. Our findings are
as follows:
\begin{itemize}
\item in Section~\ref{sec:weak_zielonka}, we prove that Zielonka's
algorithm solves weak games in polynomial time.

\item in Section~\ref{sec:dull_and_solitaire_zielonka}, we demonstrate
that, somewhat surprisingly, Zielonka's algorithm is exponential
on dull games and solitaire games.

\end{itemize}
The exponential lower bounds we obtain utilise a family of dull, solitaire
games $\solitaire[k]$ with $3k$ vertices on which the algorithm requires
$2^k$ iterations, allowing us to establish a lower bound of $\Omega(2^{n/3})$. 
This lower bound improves on previously documented
lower bounds for this algorithm (\eg, 
in~\cite{Fri:11} a lower bound of $\Omega(1.6^{n/5})$ is established).

In addition to the above complexity results we investigate whether the
most common improvement of the algorithm permits it to run in polynomial
time for all three special classes of games. That is, we prove in
Section~\ref{sec:polynomial_scc} that integrating Zielonka's algorithm
in a \emph{strongly connected component} decomposition algorithm, as
suggested in~\cite{Jur:00,FL:09}, permits solving all three classes in
polynomial time. We analyse the complexity of the resulting algorithm
for these three classes, showing that the optimised algorithm runs
in $\mc{O}(n \cdot (n + m))$ for weak, dull and (nested) solitaire
games.  Note that these worst-case complexities are slightly worse than
those for the dedicated algorithms, but that the applicability of the
algorithm remains universal; \eg, it is capable of solving arbitrary
nestings of dull and solitaire games, and it does not depend on dedicated
algorithms for detecting whether the game is special.

The optimised algorithm still requires exponential time on
non-special games. For instance, Friedmann's games are resilient
to all known optimisations.  Drawing inspiration from our family of
games $\solitaire[k]$ and the games of~\cite{Fri:11}, we define a new
family of games $\whitegame[k]$ containing $3k$ vertices, that is also
resilient to all known optimisations and requires $2^k$ iterations
of the algorithm. This again allows us to establish a lower bound of
$\Omega(2^{n/3})$, also improving on the lower bound established by
Friedmann in~\cite{Fri:11}.  We experimentally compare the running time 
of the optimised
algorithm on our games to those of Friedmann.

\paragraph{Outline.} Before we present our results, we briefly describe 
parity games
in Section~\ref{sec:parity_games} and Zielonka's algorithm in
Section~\ref{sec:zielonka}. Our runtime analysis of Zielonka's original 
algorithm on special games is presented in Section~\ref{sec:special_games}.
We prove that an optimisation of the algorithm runs in polynomial time on
special games in Section~\ref{sec:polynomial_scc}, and we prove that, in
general, the optimisation's complexity is $\Omega(2^{n/3})$ in
Section~\ref{sec:tightness}. In Section~\ref{sec:conclusions}, we
wrap up with some conclusions.

\section{Parity Games}
\label{sec:parity_games}

A parity game is an infinite duration game, played by players \emph{odd},
denoted by $\odd$ and \emph{even}, denoted by $\even$, on a directed,
finite graph. The game is formally defined as follows.

\begin{definition}
A pseudo parity game is a tuple $(V, E, \priosym, (V_\even,V_\odd))$, where
\begin{itemize}
\item $V$ is a finite set of vertices, partitioned in a set $V_\even$ of
vertices owned by player $\even$, and a set of vertices $V_\odd$ 
owned by player $\odd$,
\item $E \subseteq V \times V$ is an edge relation,
\item $\priosym \oftype V \to \nat$ is a priority function that assigns
priorities to vertices, players.
\end{itemize}
We write $v \to w$ iff $(v,w) \in E$.  A pseudo parity game is a
\emph{parity game} if the edge relation is total; \ie for each $v \in V$
there is at least one $w \in V$ such that $(v,w) \in E$.

\end{definition}

We depict (pseudo) parity games as graphs in which diamond-shaped vertices
represent vertices owned by player $\even$ and box-shaped vertices
represent vertices owned by player $\odd$. Priorities, associated with
vertices, are written inside vertices.

For a given (pseudo) parity game, we are often interested in the subgame 
that is obtained by restricting the game to a given set of vertices in
some way. Formally, we define such subgames as follows.

\begin{definition} Let $G = (V,E, \priosym, (V_\even, V_\odd))$ be a (pseudo)
parity game and let $A \subseteq V$ be an arbitrary non-empty
set. The (pseudo) parity game $G \cap A$ is the tuple $(A,
E \cap (A \times A), \priosym|_A, (V_\even \cap A, V_\odd \cap A))$.
The (pseudo) parity game
$G \setminus A$ is defined as the game $G \cap (V \setminus A)$.

\end{definition}

Throughout this section, assume that $G = (V,E,\priosym, (V_\even, V_\odd))$ 
is an arbitrary pseudo parity game.
Note that in general, whenever $G$ is a \emph{parity game} then it is not
necessarily the case that the pseudo parity games $G \setminus A$ and
$G \cap A$ are again parity games, as totality may not be preserved.
However, in what follows, we only consider constructs in which these
operations guarantee that totality \emph{is} preserved.

The game $G$ is said to be \emph{strongly connected}, see~\cite{Tar:72},
if for all pairs of vertices $v,w \in V$, we have $v \to^* w$ and
$w \to^* v$, where $\to^*$ denotes the transitive closure of $\to$. A
\emph{strongly connected component} of $G$ is a maximal set $C \subseteq
V$ for which $G \cap C$ is strongly connected.

\begin{lemma} Let $C \subseteq V$ be a strongly connected component. If
$G$ is a parity game, then so is
$G \cap C$.
\end{lemma}

Henceforth, we assume that $G$ is a parity game (\ie its edge relation is
total), and $\i$ denotes an arbitrary player. We write $\pnot{\i}$ for $\i$'s
opponent; \ie $\pnot{\even}=\odd$ and $\pnot{\odd}=\even$.
A sequence of vertices $v_1, \ldots, v_n$ is a \emph{path} if $v_m \to
v_{m+1}$ for all $1 \leq m < n$.  Infinite paths are defined in a similar
manner. We write $p_n$ to denote the $n^\textrm{th}$ vertex in a path $p$.

A game starts by placing a token on a vertex $v \in V$.  Players move the
token indefinitely according to a simple rule: if the token is
on some vertex $v \in V_{\i}$, player $\i$ gets to move the token to an
adjacent vertex.  The choice where to move the token next is determined
by a partial function $\sigma \oftype V^+ \to V$, called a \emph{strategy}.
Formally, a strategy $\sigma$ for player $\i$ is a function satisfying
that whenever it is defined for a finite path $v_1, \ldots, v_n$, we
have $\sigma(v_1,\ldots,v_n) \in \{w \in V ~|~ v \to w\}$ and
$v_n \in V_{\i}$. We say that an infinite path $v_1, v_2, \ldots$
is \emph{consistent} with a strategy $\sigma$ for player $\i$ if
for all finite prefixes $v_1,\ldots,v_n$ for which $\sigma$ is defined, we have
$\sigma(v_1,\ldots,v_n) = v_{n+1}$. An infinite path induced by strategies
for both players is called a \emph{play}.

The winner of a play is determined by the \emph{parity} of the \emph{highest}
priority that occurs infinitely often on it: player $\even$ wins if,
and only if this priority is even. That is, we here consider \emph{max} parity
games. Note that, alternatively, one could
demand that the \emph{lowest} priority that occurs infinitely often
along a play determines the winner; such games would be \emph{min} parity
games. 

A strategy $\sigma$ for player $\i$ is \emph{winning} from a vertex $v$
if and only if $\i$ is the winner of every play starting in $v$ that is
consistent with $\sigma$. A vertex is won by $\i$ if $\i$ has a winning
strategy from that vertex. Note that parity games are \emph{positionally
determined}, meaning that a vertex is won by player $\i$ if $\i$ has a
winning \emph{positional strategy}: a strategy that determines where
to move the token next based solely on the vertex on which the token
currently resides. Such strategies can be represented by a function
$\sigma\oftype V_{\i} \to V$. A consequence of positional determinacy is
that vertices are won by exactly one player~\cite{EJ:91}.  \emph{Solving}
a parity game essentially is computing the partition $(W_\even,W_\odd)$
of $V$ of vertices won by player $\even$ and player $\odd$, respectively.
We say that a game $G$ is a \emph{paradise} for player $\i$
if all vertices in $G$ are
won by \i.


\paragraph{Special games.}

Parity games pop up in a variety of practical problems. These
include model checking problems for fixed point logics~\cite{EJ:91},
behavioural equivalence checking problems~\cite{CPPW:07} and satisfiability
and synthesis
problems~\cite{AVW:03}.  In many cases, the parity games underlying
such problems are \emph{special games}: parity games with a particular
structure. We here consider three such special games: \emph{weak, dull}
and \emph{nested solitaire} games; these classes have previously been studied
in the literature, see \eg~\cite{BG:04} and the references therein. The
definitions that we present here are taken from~\cite{BG:04}.\\

Weak games are game graphs in which the priorities along paths are
monotonically descending (this is not to be confused with parity games
with \emph{weak parity} conditions). That is, for each pair of vertices
$v,w$ in the graph, if $v \to w$, then $\prio{v} \ge \prio{w}$.
Such games correspond naturally
to model checking problems for the alternation-free modal $\mu$-calculus.

\begin{definition} A parity game is \emph{weak} if the priorities 
along all paths are descending.

\end{definition}

Dedicated solvers for weak games can solve these in $\mc{O}(|V|+|E|)$.
The algorithm that does so is rather straightforward. Since parity
games are total, the set $L$ of vertices with lowest priorities $m$ are
immediately won by player $\even$ iff $m$ is even.  Any vertex in the
game that can be \emph{forced} to $L$ by the player winning $L$ can then
be removed from the game; technically, this is achieved by computing the
\emph{attractor set} (see the next section) into $L$. What remains is
another weak parity game which can be solved following the same steps
until no vertex is left.\\

Weak games are closely related to dull games: the latter are game graphs
in which all \emph{basic cycles} in the graph are disjoint.  A basic
cycle is a finite path $v_1,\ldots,v_n$ for which $v_n \to v_1$ and no
vertex $v_i$ occurs twice on the path. An \emph{even} cycle is a cycle
in which the dominating (\ie highest) priority is even; the cycle is an
\emph{odd} cycle if the dominating priority occurring on the cycle is odd.

\begin{definition} A parity game is \emph{dull} if even cycles and odd
cycles are disjoint.

\end{definition}

Note that every weak game is dull; every dull game, on the other hand,
can be converted in linear time to a weak game by changing priorities
only. This is achieved by assigning a priority that has the same parity
as the highest priority present in a strongly connected component to all
vertices in that component.  This is harmless as each strongly connected
component is either entirely even dominated or entirely odd dominated:
if not, even cycles and odd cycles would overlap, contradicting the
fact that the game is dull.  Working bottom-up, it is straightforward to
ensure that the priorities along the paths in the game are descending.
As a result, dull games can also be solved in $\mc{O}(|V| + |E|)$ using
the same algorithm as that for solving weak games. Dull games, too, can
be obtained from alternation-free $\mu$-calculus model checking problems,
and they correspond naturally to the alternation-free fragment of LFP,
see~\cite{BG:01}.\\

Solitaire games are games in which only one of the two players gets to
make non-trivial choices where to play the token next; nested solitaire
games generalise solitaire games to games in which both players may make 
non-trivial
moves, but the interactions between both players is still restricted. Such
games arise from model checking problems for the fragment $L_2$ of the
modal $\mu$-calculus, see~\cite{EJS:01}, and they correspond with the
solitaire fragment of LFP~\cite{BG:04}.  

\begin{definition} A parity game is \emph{solitaire} if all non-trivial
moves are made by a single player. The game is \emph{nested solitaire}
if each strongly connected component induces a solitaire game.

\end{definition}

Nested solitaire games can be solved in $\mc{O}(\log(d) \cdot (|V| +
|E|))$, see~\cite{GK:05}, although most implementations use a somewhat
less optimal implementation that runs in $\mc{O}(d \cdot (|V| + |E|))$,
see~\cite{BG:04}. The latter algorithm works by computing the strongly
connected components of a graph and start searching for an even cycle if
all non-trivial moves in the component are made by player $\even$ and
an odd cycle otherwise.
Computing whether there is an even cycle
(resp.\ odd cycle) can be done in $\mc{O}(\log(d) \cdot (|V| + |E|))$
using the techniques of~\cite{KKV:01} or, in $\mc{O}(d \cdot (|V| +
|E|))$ by repeatedly conducting a depth-first search, starting at the
lowest even priority in the component. Clearly, in a component where
only player $\even$ gets to make non-trivial moves,
all vertices are won by
player $\even$ iff an even cycle is found. Iteratively solving the final
strongly connected component and removing it together with the attractor
for the winner of this component solves entire nested solitaire games.

\section{Zielonka's Recursive Algorithm}
\label{sec:zielonka}

Throughout this section, we assume $G$ is a fixed parity game $(V,E,
\priosym, (V_\even, V_\odd))$, and $\i$ is an arbitrary player. 

Zielonka's algorithm for solving parity games, listed as
Algorithm~\ref{alg:zielonka}, is a divide and conquer algorithm. It
constructs winning regions for both players out of the solution of
subgames with fewer different priorities and fewer vertices. It removes
the vertices with the highest priority from the game, together with all
vertices \emph{attracted} to this set of vertices. Attractor sets are
formally defined as follows.
\begin{definition} The $\i$-\emph{attractor} into a set $U \subseteq V$,
denoted $\attr{\i}{U}$, is defined
inductively as follows:
\[
\begin{array}{lcl}
\attr[0]{\i}{U} & = & U \\
\attr[n+1]{\i}{U} & = & \attr[n]{\i}{U} \\
      & \cup & \{u \in V_{\i} ~|~ \exists v \in \attr[n]{\i}{U}:~ u \to v\} \\
      & \cup & \{u \in V_{\pnot{\i}} ~|~ \forall v \in V: u \to v \implies v \in \attr[n]{\i}{U}\} \\
\attr{\i}{U} & = & \bigcup\limits_{i \ge 0} \attr[i]{\i}{U}
\end{array}
\]
If needed for clarity, we write $\attr[G]{\i}{U}$ to indicate 
that the $\i$-attractor is computed in game graph $G$.
\end{definition}

The lemma below states that whenever attractor sets are removed
from a parity game, totality is preserved.

\begin{lemma} Let $A = \attr{\i}{U} \subseteq V$ be an arbitrary attractor
set. If $G$ is a parity game, then so is $G \setminus A$.

\end{lemma}

The correctness of 
Zielonka's algorithm hinges on the fact that higher priorities in the game
dominate lower priorities, and that any forced revisit of these higher
priorities is beneficial to the player aligning with the parity of the
priority.  For a detailed explanation of the algorithm and proof of
its correctness, we refer to~\cite{Zie:98,Fri:11}.

\begin{algorithm}[ht]
\caption{Zielonka's Algorithm}
\label{alg:zielonka}
\begin{algorithmic}[1]
\Function{Zielonka}{G}
\If{$V = \emptyset$}
     \State $(W_\even,W_\odd) \gets (\emptyset,\emptyset)$
\Else
     \State $m \gets \max\{ \prio{v} ~|~ v \in V\}$
     \State \textbf{if} $m\!\!\!\mod~2 = 0$ \textbf{then} $p \gets \even$
            \textbf{else} $p \gets \odd$ \textbf{end if} 
     \State $U \gets \{v \in V ~|~ \prio{v} = m\}$
     \State $A \gets \attr{p}{U}$
     \State $(W'_\even, W'_\odd) \gets \Call{Zielonka}{G\setminus A}$
     \If{$W'_{\pnot{p}} = \emptyset$}
           \State $(W_{p}, W_{\pnot{p}}) \gets (A \cup W'_{p}, \emptyset)$
     \Else
           \State $B \gets \attr{\pnot{p}}{W'_{\pnot{p}}}$
           \State $(W'_\even, W'_\odd) \gets \Call{Zielonka}{G\setminus B}$
           \State $(W_{p},W_{\pnot{p}}) \gets (W'_{p},W'_{\pnot{p}} \cup B)$
     \EndIf
\EndIf
\State \Return $(W_\even, W_\odd)$
\EndFunction
\end{algorithmic}
\end{algorithm}

\section{Solving Special Games}
\label{sec:special_games}

Zielonka's algorithm is quite competitive on parity games that stem
from practical verification problems~\cite{FL:09,Kei:09}, often beating
algorithms with better worst-case running time. While Zielonka's
original algorithm is known to run in exponential time on games defined
by Friedmann~\cite{Fri:11}, its behaviour on special parity games
has never before been studied. It might just be the case that this
algorithm is particularly apt to solve such games. We partly confirm
this hypothesis in Section~\ref{sec:weak_zielonka} by proving that
the algorithm indeed runs in polynomial time on weak games. Somewhat
surprisingly, however, we also establish that Zielonka's algorithm
performs poorly when solving dull and (nested) solitaire games, see
Section~\ref{sec:dull_and_solitaire_zielonka}.

\subsection{Weak Games}
\label{sec:weak_zielonka}

We start with a crucial observation ---namely, that for weak games,
\ziel solves a paradise in polynomial time--- which permits us to prove that
solving weak games can be done in polynomial time using \ziel. The proof
of the latter, formalised as Proposition~\ref{prop:paradise}, depends on 
three observations, which we first prove in isolation in the following
lemma.
\begin{lemma}
\label{lem:paradise} Let $G = (V,E,\priosym,(V_\even,V_\odd))$ be a weak
parity game. Suppose $G$ is a paradise for player $\i$; \ie, $G$ is
won entirely by $\i$. Then \ziel, applied to $G$, has the following properties:
\begin{enumerate}
\item in the first recursive call in line $9$, the argument $G \setminus A$ is also a paradise for player $\i$.
\item if the second recursive call (line $14$) is reached, then its argument ($G \setminus B$) is the empty set.
\item edges that are used in the computation of attractor sets (lines $8$ and $13$) are not considered in subroutines.
\end{enumerate}

\end{lemma}

\begin{proof} We prove all three statements below.
 \begin{enumerate}
  \item Observe that $A = \attr{p}{U} = U$, since, in a weak game,
  no vertex with lower priority has an edge to a vertex in $U$. In
  particular, the subgame $G \setminus A$ is $\pnot{\i}$-closed, and
  hence must be won entirely by $\i$, if $G$ is a $\i$-paradise.

  \item The second recursive call can be invoked only if $W'_{\pnot{p}}
  \neq \emptyset$. From the above considerations we know that this
  implies $\pnot{p} = \i$, and $W'_{\pnot{p}} = G \setminus A$ is a
  paradise for $\i$. We also have $G = W'_{\pnot{p}} \cup A$. Since
  every game staying in $A$ would be losing for $\i$, it must be the
  case that $A \subseteq Attr_{\pnot{p}}(W'_{\pnot{p}})$. But then $B
  = Attr_{\pnot{p}}(W'_{\pnot{p}}) = G$, and hence $G \setminus B =
  \emptyset$.

  \item Edges that are considered in the computation of both $\attr{p}{U}$
  (line $8$) and $\attr{\pnot{p}}{W'_{\pnot{p}}}$ have sources only in
  $U$; since no vertices from $U$ are included in the subgame
  considered in the first recursive call, and the second call can only
  take the empty set as an argument. Therefore, these edges will not be
  considered in the subroutines.

  \end{enumerate}
\end{proof}

\begin{proposition}
\label{prop:paradise} Let $G = (V,E,\priosym,(V_\even,V_\odd))$ be a weak
parity game. Suppose $G$ is a paradise for player $\i$; \ie, $G$ is
won entirely by $\i$. Then \ziel runs in $\mc{O}(|V| + |E|)$.

\end{proposition}
\begin{proof}
We analyse the running time $T(k)$ of \ziel when it is called on a
subgame $G_k$ of $G$ with exactly $k$ priorities. Let $v_k$ denote
the number of nodes with the highest priority in $G_k$, and with $e_k$
the number of edges that are considered in the attractor computations
(lines $8$ and $13$) on $G_k$.

If we assume that the representation of the game has some built-in functionality that allows us to inspect the
nodes in the order of priority, then the time required to execute the specific lines of the procedure can be bounded as follows:
\begin{itemize}
 \item line $7$: $c \cdot v_k$ for some constant $c$
 \item lines $8$ and $13$ in total: $c \cdot e_k$ for some constant $c$
 \item the remaining lines: $z$ for some constant $z \in \nat$
\end{itemize}
We obtain:
\[
\begin{array}{lll}
 T(k) & \leq & c \cdot (v_i + e_i + z) + T(k-1)\\ 
 T(k) & \leq & \sum\limits_{i=1}^{k}\ c \cdot (v_i + e_i + z)\\
 T(k) & \leq & c \cdot ( \sum\limits_{i=1}^{k} v_i + \sum\limits_{i=1}^{k} e_i + \sum\limits_{i=1}^{k} z) \\ 
\end{array}
\]
Let $d$ denote the total number of priorities occurring in
$G$. Observe that from Lemma~\ref{lem:paradise}, we have:
\[\sum\limits_{i=1}^{d} v_i = |V| \text{ and }
\sum\limits_{i=1}^{d} e_i = |E|
\]
The
total execution time of \ziel on $G$ can be bounded by:
\[
 T(d,V,E) \,\leq\, c \cdot (|V| + |E| + {\cal O}(d))\\
 \]
Hence we obtain
$
 T(d,V,E) \,=\, {\cal O} (|V| + |E|)
$.
\end{proof}

The above proposition is used in our main theorem below to prove
that solving weak games using \ziel can be done in polynomial time:
each second recursive call to \ziel will effectively be issued on
a paradise or an empty game. By proposition~\ref{prop:paradise},
we know that \ziel will solve a paradise in linear time. 

\begin{theorem}
\ziel requires $\mc{O}(d \cdot (|V| + |E|))$ to solve weak games with
$d$ different priorities, $|V|$ vertices and $|E|$ edges.
\end{theorem}
\begin{proof}
The key observation is that \ziel, upon entering the second recursive
call in line 14 is invoked on a game that is a paradise. Consider the
set of vertices $V \setminus B$ of the game $G$ at that point. It
contains the entire set $W'_p$, and possibly a subset of $U$. Now,
if player $\pnot{p}$ could force a play in a node $v \in W'_p$ to
$W'_{\pnot{p}}$, it could be done only via set $U$. But this would
violate the weakness property. Player $p$ has a winning strategy on
$V \setminus B$, which combines the existing strategy on $W'_{p}$ and
if necessary any strategy on $U$ (because whenever a play visits $U$
infinitely often, it is won by $p$). Thus, the game $G \setminus B$
that is then considered is a $p$-paradise.

As a result, by Proposition~\ref{prop:paradise}, the game $G \setminus B$
is solved in $\mc{O}(|V| + |E|)$. Based on these observations, we obtain the
following recurrence for \ziel:
\[
\begin{array}{ll}
T(0,V,E) & \le \mc{O}(1) \\
T(d+1,V,E) & \le T(d,V,E) + \mc{O}(|V|) + \mc{O}(|E|)
\end{array}
\]
Thus, a non-trivial upper bound on the complexity is $\mc{O}(d \cdot
(|V| + |E|))$.
\end{proof}
Next, we show this bound is tight.  Consider the family of parity
games $\weak[n] = (V^n, E^v, \priosym^n, (V^n_\even, V^n_\odd))$, where
priorities and edges are defined in Table~\ref{tab:weak_family} and $V^n$
is defined as $V^n = \{v_1, \ldots, v_{2n}, u_0, u_1\}$.

\begin{table}[h]
\centering
\small
\caption{The family $\weak$ of games; $1 \le i \le n$.}
\label{tab:weak_family}
\begin{tabular}{c|c|c|c}
Vertex & Player &  Priority & Successors \\
\hline
$v_i$ & $\even$ & $i+2$ & $\{v_{i-1}, v_{n+i} \} \cup \{u_0 ~|~ i = 1 \}$\\
$v_{n+i}$ & $\odd$ & $i+2$ & $\{v_{i}, v_{n+i-1} \} \cup \{u_1~|~i=1\}$\\
$u_0$ & $\even$ & $0$ & $\{u_0\}$ \\
$u_1$ & $\odd$ & $1$ & $\{u_1 \}$\\
\end{tabular}
\end{table}
\noindent
The game $\weak[4]$ is depicted in Figure~\ref{fig:weak_5}.
\begin{figure}[h]
\centering
\begin{tikzpicture}[>=stealth']
\tikzstyle{every node}=[draw, inner sep=1pt, outer sep=1pt,minimum size=8pt];
  \node[shape=diamond,label=above:{$v_4$}] (y5)               {\scriptsize 6};
  \node[shape=diamond,label=above:{$v_3$}] (y4) [right of=y5,xshift=30pt] {\scriptsize 5};
  \node[shape=diamond,label=above:{$v_2$}] (y3) [right of=y4,xshift=30pt] {\scriptsize 4};
  \node[shape=diamond,label=above:{$v_1$}] (y2) [right of=y3,xshift=30pt] {\scriptsize 3};
  \node[shape=diamond,label=above:{$u_0$}] (y1) [right of=y2,xshift=30pt] {\scriptsize 0};

  \node[shape=rectangle,label=below:{$v_8$}] (x5) [below of=y5] {\scriptsize 6};
  \node[shape=rectangle,label=below:{$v_7$}] (x4) [below of=y4] {\scriptsize 5};
  \node[shape=rectangle,label=below:{$v_6$}] (x3) [below of=y3] {\scriptsize 4};
  \node[shape=rectangle,label=below:{$v_5$}] (x2) [below of=y2] {\scriptsize 3};
  \node[shape=rectangle,label=below:{$u_1$}] (x1) [below of=y1] {\scriptsize 1};

\path[->]
  (y5) edge (y4) edge[bend left] (x5)
  (x5) edge (x4) edge[bend left] (y5)
  (y4) edge (y3) edge[bend left] (x4)
  (x4) edge (x3) edge[bend left] (y4)
  (y3) edge (y2) edge[bend left] (x3)
  (x3) edge (x2) edge[bend left] (y3)
  (y2) edge (y1) edge[bend left] (x2)
  (x2) edge (x1) edge[bend left] (y2)
  (x1) edge[loop right] (x1)
  (y1) edge[loop right] (y1)
;
\end{tikzpicture}
\caption{The game $\weak[4]$.}
\label{fig:weak_5}
\end{figure}
The family $\weak$ has the following characteristics.
\begin{proposition} The game $\weak[n]$ is of size linear in $n$; \ie,
$|\weak[n]| = \mc{O}(n)$, it contains $2n + 2$ vertices, $4n +2$ edges and
$n+2$ different priorities. Moreover, the game $\weak[n]$ is a weak game.
\end{proposition}
\begin{lemma} In the game $\weak[n]$, vertices $\{u_0,v_1,\ldots,v_n\}$ are
won by player $\even$, whereas vertices $\{u_1,v_{n+1},\ldots,v_{2n}\}$
are won by player $\odd$.
\end{lemma}
\begin{proof}
Follows from the fact that, for $0 \le j < n-1$, the strategy $v_{n-j}
\to v_{n-j-1}$, $v_1 \to u_0$ and $u_0 \to u_0$ is
winning for player $\even$ for the set of vertices $\{u_0,v_11,\ldots,v_n\}$
and the strategy $v_{2n-j} \to v_{2n-j-1}$, $v_{n+1} \to u_1$ and $u_1
\to u_1$ is winning for player $\odd$ from the set of
vertices $\{u_1,v_{n+1},\ldots,v_{2n}\}$.
\end{proof}
We next analyse the runtime of Zielonka's algorithm on the family $\weak$.
Let $a_n$ be defined through the following recurrence relation:
\[
\begin{array}{ll}
a_0 & = 1 \\
a_{n+1} & = a_n + n + 1
\end{array}
\]
Observe that the function $\frac{1}{2}n^2$ approximates $a_n$ from below.
The proposition below states that solving the family $\weak$ of weak
parity games requires a quadratic number of recursions.
\begin{proposition} Solving $\weak[n]$, for $n > 0$, requires at least
$a_n$ calls to \ziel.
\end{proposition}
\begin{proof}
Follows from the observation that each game $\weak[n+1]$ involves:
\begin{enumerate}
\item a first recursive call to \ziel for solving the game $\weak[n]$.
\item a second recursive call to \ziel for solving either $\weak[n]\setminus
\{v_{2n},\ldots,v_{n+1}, u_1\}$ or $\weak[n]\setminus
\{v_n,\ldots,v_1,u_0\}$; both require $n+1$ recursive calls to \ziel.
\end{enumerate}
\end{proof}
\begin{theorem} Solving weak games using \ziel requires
$\Theta(d \cdot (|V| + |E|))$.
\end{theorem}
Note that this complexity is a factor $d$ worse than that of the
dedicated algorithm. For practical problems such as when solving parity
games that come from model checking problems $d$ is relatively
small; we expect that for such cases, the difference between the
dedicated algorithm and Zielonka's algorithm to be small.

\subsection{Dull and Nested Solitaire Games}
\label{sec:dull_and_solitaire_zielonka}

We next prove that dull games and (nested) solitaire require exponential 
time to solve using
\ziel. Given that dull games can be converted to weak games in linear
time, and that Zielonka solves weak games in polynomial time, this may
be unexpected.  

Our focus is on solitaire games first. We construct a family of 
parity games $\solitaire[n] =
(V^n,E^n,\priosym^n, (V^n_\even,V^n_\odd))$ with vertices
$V^n = \{v_0, \ldots, v_{2n-1}, u_1, \ldots u_n\}$.
All vertices belong to player \even; that is, $V^n_\even = V^n$ and
$V^n_\odd = \emptyset$. The priorities and the edges are described by
Table~\ref{tab:grey_game}.
\begin{table}[h!]
\centering
\small
\caption{The family $\solitaire$ of games; $1 \le i < 2n, 1 \le j \le n$.}
\label{tab:grey_game}
\begin{tabular}{c|c|c}
Vertex & Priority & Successors \\
\hline
$v_i$ & $i+2$ & $\{v_{i-1} \}$\\
$v_0$ & $2$ & $\{v_0\}$ \\
$u_j$ & $1$ & $\{u_j, v_{2j-1} \}$\\
\end{tabular}

\end{table}
\begin{proposition} The game $\solitaire[n]$ is of size linear in $n$;
\ie $|\solitaire[n]| = \mc{O}(n)$, it has $3n$ vertices, $4n$ edges and
$2n+1$ different priorities. Moreover, the game $\solitaire[n]$ is a
(nested) solitaire game.

\end{proposition}
%
%

\begin{figure}[b]
\centering
\begin{tikzpicture}[>=stealth']
\tikzstyle{every node}=[draw, inner sep=1pt, minimum size=8pt];
  \node[shape=diamond, label=above:$v_5$] (v1) { \scriptsize 7};
  \node[shape=diamond, label=above:$v_4$] (v2) [right of=v1,xshift=30pt] { \scriptsize 6};
  \node[shape=diamond, label=above:$v_3$] (v3) [right of=v2,xshift=30pt] { \scriptsize 5};
  \node[shape=diamond, label=above:$v_2$] (v4) [right of=v3,xshift=30pt] { \scriptsize 4};
  \node[shape=diamond, label=above:$v_1$] (v5) [right of=v4,xshift=30pt] { \scriptsize 3};
  \node[shape=diamond, label=above:$v_0$] (v6) [right of=v5,xshift=30pt] { \scriptsize 2};

  \node[shape=diamond, label=left:$u_3$] (u1) [below of=v1] { \scriptsize 1};
  \node[shape=diamond, label=left:$u_2$] (u2) [below of=v3] { \scriptsize 1};
  \node[shape=diamond, label=left:$u_1$] (u3) [below of=v5] { \scriptsize 1};

  \path[->]
  (v1) edge (v2)
  (v2) edge (v3)
  (v3) edge (v4)
  (v4) edge (v5)
  (v5) edge (v6)
  (v6) edge[loop right] (v6)
  (u1) edge (v1) edge [loop right] (u1)
  (u2) edge (v3) edge [loop right] (u2)
  (u3) edge (v5) edge [loop right] (u3)
   ;

\end{tikzpicture}
\caption{The game $\solitaire[3]$.}
\label{fig:solitaire_3}
\end{figure}

\noindent
The game $\solitaire[3]$ is depicted in Figure~\ref{fig:solitaire_3}.
Observe that in this game, vertex $v_5$ has the maximal priority and that
this priority is odd. This means that Zielonka's algorithm will compute
the odd-attractor to $v_5$ in line 8 of the algorithm, \ie
$\attr{\odd}{\{v_{5}\}} = \{v_{5}\}$. We can generalise this observation
for arbitrary game 
$\solitaire[n]$: in such a game,
$\attr{\odd}{\{v_{2n-1}\}} = \{v_{2n-1}\}$.  Henceforth, we denote the
subgame $\solitaire[n] \setminus \{v_{2n-1}\}$ by $\solitaire[n,-]$.

\begin{lemma}
\label{lem:solutions_solitaire}
The game $\solitaire[n]$ is won by 
player \even.
In the game $\solitaire[n,-]$, all vertices except for vertex $u_n$,
are won by player \even.
\end{lemma}
\begin{proof} The fact that $\solitaire[n]$ is won by player \even follows
immediately from the strategy $\sigma \oftype V^n \to V^n$, defined
as $\sigma(v_i) = v_{i-1}$ for all $1 \le i < 2n$, $\sigma(u_i) = v_{2i -1}$
for all $i \le n$ and $\sigma(v_0) = v_0$, which is winning for
player \even.  For the game $\solitaire[n,-]$, a strategy $\sigma'$ can
be used that is as strategy $\sigma$ for all vertices $v \not= u_n$;
for vertex $u_n$, we are forced to choose $\sigma'(u_n) = u_n$, since
$u_n$ is the sole successor of $u_n$ in $\solitaire[n,-]$. Since the priority
of $u_n$ is odd, the vertex is won by player $\odd$.
\end{proof}
We now proceed to the runtime of Zielonka's algorithm on the family $\solitaire$.

\myomit{Let $a_n$ be defined through the following recurrence relation:
\[
\begin{array}{ll}
a_0 & = 1\\
a_{n+1} & = 2*a_n \\
\end{array}
\]
Note that $a_n = 2^n$; that is, $a_n$ is exponential in $n$.}
\begin{proposition}
\label{prop:solitaire-calls}
Solving $\solitaire[n]$, for $n > 0$, requires at least $2^n$ calls to
\ziel.
\end{proposition}
\begin{proof}
Solving the game $\solitaire[1]$ requires at least one call to \ziel.

Consider the game $\solitaire[n]$, for $n > 1$.  Observe that \ziel is
invoked recursively on the game $\solitaire[n,-]$ in the first recursion
on line 9.  We focus on solving the latter game.

The vertex with the
highest priority in $\solitaire[n,-]$ is $v_{2(n-1)}$. Observe that $A =
\attr{\odd}{\{v_{2(n-1)}\}} = \{v_{2(n-1)}\}$. Note that
$\solitaire[n,-] \setminus A$ contains $\solitaire[n-1]$ as a separate subgame.
The first recursive call in solving $\solitaire[n,-]$ will therefore also solve
the subgame $\solitaire[n-1]$.

Next, observe that $u_{n}$ (and $u_{n}$ alone) is won by player
$\odd$, see Lemma~\ref{lem:solutions_solitaire}. We therefore need
to compute $B = \attr{\even}{\{u_n\}} = \{u_n\}$. Now, note that
$\solitaire[n,-] \setminus B$ subsumes the subgame $\solitaire[n-1]$,
which is a separate game in $\solitaire[n,-] \setminus B$.
Therefore, also the second recursive call to \ziel involves solving the subgame
$\solitaire[n-1]$.
\end{proof}
The lower bound on the number of iterations for \ziel is thus exponential 
in the number of vertices.
\begin{theorem}
Solving (nested) solitaire games using \ziel requires $\Omega(2^{|V|/3})$.
\end{theorem}
We note that this improves on the bounds of $\Omega(1.6^{|V|/5})$ established by Friedmann. Being
structurally more complex, however, his games are robust to typical
(currently known) improvements to Zielonka's algorithm such as the
one presented in the next section (although this is not mentioned or
proved in~\cite{Fri:11}). Still, we feel that the
simplicity of our family $\solitaire$ fosters a better understanding
of the algorithm.

Observe that the family $\solitaire$ is also a family of dull games. As
a result, we immediately have the following theorem.
\begin{theorem}
Solving dull games using \ziel requires $\Omega(2^{|V|/3})$.
\end{theorem}

\section{Recursively Solving Special Games in Polynomial Time}
\label{sec:polynomial_scc}

The $\solitaire$ family of games of the previous section are easily
solved when preprocessing the games using priority propagation and
self-loop elimination. However, it is straightforward to make the
family robust to such heuristics by duplicating the vertices that have
odd priority, effectively creating odd loops that are not detected by
such preprocessing steps.  In a similar vein, the commonly suggested
optimisation to use a strongly connected component decomposition as a
preprocessing step can be shown to be insufficient to solve (nested)
solitaire games. The family $\solitaire$ can easily be made robust to
this preprocessing step: by adding edges from $v_0$ to all $u_i$, each
game in $\solitaire$ becomes a single SCC.

In this section, we investigate the complexity of a tight integration of
a strongly connected component decomposition and Zielonka's algorithm,
as suggested by \eg~\cite{Jur:00,FL:09}.  By decomposing the game each time
Zielonka is invoked, large SCCs are broken down in smaller SCCs,
potentially increasing the effectiveness of the optimisation. The
resulting algorithm is listed as Algorithm~\ref{alg:optimised_zielonka}.
\begin{algorithm}[ht]
\caption{Optimised Zielonka's Algorithm}
\label{alg:optimised_zielonka}
\begin{algorithmic}[1]
\Function{Zielonka\_SCC}{G}
\State $(W^G_\even, W^G_\odd) \gets (\emptyset, \emptyset)$
\If{$V \not= \emptyset$}
     \State $\mc{S} := \Call{SCC\_Graph\_Decomposition}{G}$
     \For{\textbf{each} final SCC $C \in \mc{S}$}
       \State $H \gets G \cap C$
       \State $m \gets \max\{ \prio{v} ~|~ v \in C\}$
       \State \textbf{if} $m\!\!\!\mod 2 = 0$ \textbf{then} $p \gets \even$
              \textbf{else} $p \gets \odd$ \textbf{end if} 
       \State $U \gets \{v \in C ~|~ \prio{v} = m\}$
       \State $A \gets \attr[H]{p}{U}$
       \State $(W'_\even, W'_\odd) \gets \Call{Zielonka\_SCC}{H\setminus A}$
       \If{$W'_{\pnot{p}} = \emptyset$}
             \State $(W_{p}, W_{\pnot{p}}) \gets (A \cup W'_{p}, \emptyset)$
       \Else
             \State $B \gets \attr[H]{\pnot{p}}{W'_{\pnot{p}}}$
             \State $(W'_\even, W'_\odd) \gets \Call{Zielonka\_SCC}{H\setminus B}$
             \State $(W_{p},W_{\pnot{p}}) \gets (W'_{p},W'_{\pnot{p}} \cup B)$
       \EndIf
     \State $(W^G_\even,W^G_\odd) \gets (W^G_\even \cup \attr[G]{\even}{W_\even} , W^G_\odd \cup \attr[G]{\odd}{W_\odd})$
     \State $\mc{S} := \Call{SCC\_Graph\_Decomposition}{G\setminus(W_\even^G \cup W_\odd^G)}$
    \EndFor
\EndIf
\State \Return $(W^G_\even, W^G_\odd)$
\EndFunction
\end{algorithmic}
\end{algorithm}

\noindent

We will need the following lemma:

\begin{lemma}
\label{lem:nosecond}
If algorithm \ref{alg:optimised_zielonka} is invoked on a game that is either dull or (nested) solitaire, then in the entire recursion tree all second recursive calls (line 16) are trivial (with empty set as an argument).
\end{lemma}

\begin{proof}

 In case of dull games, since game $H$ is a connected component, each of its subgames is won by player corresponding to $m \!\!\!\mod 2$, namely $p$. Hence after the line 11 is executed, we obtain $W'_{p} = H \setminus A$ and $W'_{\pnot{p}} = \emptyset$. The second recursive call will therefore never be invoked.
 
 Now assume that the game is solitaire and owned by player $q$. If $p=q$, then
 $\attr[H]{p}{U} = H$ (the game is $p$-owned, and strongly connected),
 and the second call is not invoked at all. Otherwise, the second
 call is invoked only if $W'_{\pnot{p}} \neq \emptyset$. But then $B =
 \attr[H]{\pnot{p}}{W'_{\pnot{p}}} = \attr[H]{q}{W'_{\pnot{p}}} = H$
 (the game is owned by $\pnot{p}$, and strongly connected), and $H
 \setminus B = \emptyset$. 
\end{proof}

We will now prove that the optimisation suffices to solve special parity games in polynomial time.
\begin{theorem}
Algorithm \ref{alg:optimised_zielonka} solves dull and (nested) solitaire games in ${\cal O}(|V| \cdot (|V| + |E|))$ time.
\end{theorem}

\begin{proof}

\def\nfor{\#\textsf{for}}


Let $\nfor(V)$ denote the total number of iterations of the \textsf{for} loop in the entire recursion tree. Observe that the total execution time of $\textsc{Zielonka\_SCC}$ can be bounded from above as follows:
\[
T(V,E) = {\cal O}(\nfor(V) \cdot (|V|+|E|))
\]
Indeed, every iteration of the loop (not counting the iterations in subroutines) contributes a maximal factor of ${\cal O}(|V|+|E|)$ running time, which results from the attractor computation and SCC decomposition.

We will use subscripts for the values of the algorithm variables in iteration $i \in \{1, \dots, k\}$, e.g. the value of variable $C$ in iteration $i$ is $C_i$. Furthermore, by $V_i$ we will denote the set of vertices in the subgame considered in the first recursive call, i.e. $V_i = C_i \setminus A_i$.

We will show that $\nfor(V) \leq |V|$. We have:
\[
\tag{*}
\begin{array}{llll}
\nfor(V) & \leq & 1 &\text{for } |V| \leq 1\\
\nfor(V) & \leq & \nfor(V_1) + \dots + \nfor(V_k) + k & \text{for }|V| > 1
\end{array}
\]

In the second inequality, $k$ is the total number of bottom SCCs considered in line 5. Each of these SCCs
may give rise to a recursive call (at most one, see Lemma \ref{lem:nosecond}). This recursive call contributes in turn $\nfor(V_i)$ iterations.

We proceed to show $\nfor(V) \leq |V|$ by induction on $|V|$. The base holds immediately from the first inequality. Now assume that $\nfor(V) \leq |V|$ for $|V| < m$.

Obviously $|C_1| + \dots + |C_k| \leq |V|$. Observe that in every iteration $i$ the set $A_i$ is nonempty, therefore $V_i < C_i$. Therefore $|V_1| + \dots + |V_k| \leq |V| - k$, or equivalently 
$|V_1| + \dots + |V_k| + k \leq |V|$. 

Applying the induction hypothesis in the right-hand side of (*) yields $\nfor(V) \leq |V_1| + \dots + |V_k| + k$, and due to the above observation we finally obtain $\nfor(V) \leq |V|$.
\end{proof}

The above upper bound is slower by a factor $V$ compared to the dedicated algorithms for solving weak and dull games. For nested solitaire games, the optimised recursive algorithm has an above upper bound comparable to that of standard dedicated algorithms for nested solitaire games when the number of different priorities is of $\mc{O}(V)$, and it is a factor $V/\log(d)$ slower compared to the most efficient algorithm for solving nested solitaire games. 

\section{A Tighter Exponential Bound for Zielonka's Optimised Algorithm}
\label{sec:tightness}

In view of the findings of the previous section, it seems beneficial
to always integrate Zielonka's recursive algorithm with SCC decomposition. Observe that the family of games
we used to establish the lower bound of
$\Omega(2^{|V|/3})$ in
Section~\ref{sec:special_games} does not permit us to prove the same
lower bound for the optimised
algorithm. As a result, the current best known lower bound for the algorithm is
still $\Omega(1.6^{|V|/5})$.  In this section, we show that the complexity of the
optimised algorithm is actually also $\Omega(2^{|V|/3})$.
The family of games we construct
is, like Friedmann's family, resilient to all optimisations we are aware of.

Let $\whitegame[n] = (V^n,E^n, \priosym^n, (V^n_\even,V^n_\odd))$,
for $n \ge 1$ be a family of parity games with set of vertices
$V^n = \{v_i,u_i,w_i ~|~ 1 \le i \le n \}$. The sets $V^n_\even$ and
$V^n_\odd$, the priority function $\priosym^n$ and the set of edges
are described by Table~\ref{tab:whitegame_family}.
We depict the game $\whitegame[4]$ in Figure~\ref{fig:whitegame}.

\begin{table}[h]
\centering
\small
\caption{The family $\whitegame$ of games; $1 \le i \le n$.}
\label{tab:whitegame_family}
\begin{tabular}{c|c|c|c}
Vertex & Player &  Priority & Successors \\
\hline
$v_i$ & $\odd$\, iff\, $i\!\!\!\mod 2 = 0$ & $i+1$ & $\{u_i\} \cup \{v_{i+1} ~|~ i < n \}$ \\
$u_i$ & $\odd$\, iff\, $i\!\!\!\mod 2 = 0$ & $i\!\!\!\mod 2$ & $\{w_i\} \cup \{v_{i+1} ~|~ i < n \}$ \\
$w_i$ & $\even$\, iff\, $i\!\!\!\mod 2 = 0$ & $i\!\!\!\mod 2$ & $\{u_i\} \cup \{w_{i-1} ~|~ 1 < i \}$ \\
\end{tabular}
\end{table}

\begin{figure}[h]
\centering
\begin{tikzpicture}[>=stealth']
\tikzstyle{every node}=[draw, inner sep=1pt, minimum size=8pt];
  \node[shape=rectangle, label=above:$v_4$] (v4) {\scriptsize 5};
  \node[shape=diamond, label=above:$v_3$]   (v3) [right of=v4,xshift=30pt] {\scriptsize 4};
  \node[shape=rectangle, label=above:$v_2$] (v2) [right of=v3,xshift=30pt] {\scriptsize 3};
  \node[shape=diamond, label=above:$v_1$]   (v1) [right of=v2,xshift=30pt] {\scriptsize 2};

  \node[shape=rectangle, label=right:$u_4$] (u4) [below of=v4]{\scriptsize 0};
  \node[shape=diamond, label=right:$u_3$]   (u3) [right of=u4,xshift=30pt] {\scriptsize 1};
  \node[shape=rectangle, label=right:$u_2$] (u2) [right of=u3,xshift=30pt] {\scriptsize 0};
  \node[shape=diamond, label=right:$u_1$]   (u1) [right of=u2,xshift=30pt] {\scriptsize 1};

  \node[shape=diamond, label=below:$w_4$] (w4) [below of=u4]{\scriptsize 0};
  \node[shape=rectangle, label=below:$w_3$]   (w3) [right of=w4,xshift=30pt] {\scriptsize 1};
  \node[shape=diamond, label=below:$w_2$] (w2) [right of=w3,xshift=30pt] {\scriptsize 0};
  \node[shape=rectangle, label=below:$w_1$]   (w1) [right of=w2,xshift=30pt] {\scriptsize 1};

\draw[->]
(v4) edge (u4)
(v3) edge (v4) edge (u3)
(v2) edge (v3) edge (u2)
(v1) edge (v2) edge (u1)

(u4) edge[bend left] (w4)
(u3) edge[bend left] (w3) edge (v4)
(u2) edge[bend left] (w2) edge (v3)
(u1) edge[bend left] (w1) edge (v2)

(w4) edge[bend left] (u4) edge (w3)
(w3) edge[bend left] (u3) edge (w2)
(w2) edge[bend left] (u2) edge (w1)
(w1) edge[bend left] (u1)

;

\end{tikzpicture}
\caption{The game $\whitegame[4]$.}
\label{fig:whitegame}
\end{figure}

\begin{proposition}
 The game $\whitegame[n]$ is won entirely by player $\even$ for even $n$ and entirely by player $\odd$ for odd $n$.
\end{proposition}
\myomit{
\begin{proof}

Consider $\whitegame[2n]$. The strategy for player $\even$ defined as $\sigma(v_{2n}) = v_{2n+1}$, $\sigma(u_n) = u_{2n+1}$, is winning. Indeed, every play will ultimately end up in one of the cycles $\{u_{2n},w_{2n}\}$ of nodes with priority $0$ (player $\odd$ will not be able to escape from the leftmost even cycle). The proof for $\whitegame[2n+1]$ and player $\odd$ is dual.
\end{proof}
}
\begin{theorem}
\label{thm:white-time}
 Solving $\whitegame[n]$ using either \ziel or \textsc{Zielonka\_SCC} requires $\Omega(2^{|V|/3})$ time.
\end{theorem}

\begin{proof}
 The proof is similar to Prop. \ref{prop:solitaire-calls}; we can show that the game $\whitegame[n]$ requires $2^n$ calls to either \ziel or \textsc{Zielonka\_SCC}. The only significant difference in case of \textsc{Zielonka\_SCC} is that the game may be potentially simplified in line 4 of Alg.\ref{alg:optimised_zielonka}. However, each game $\whitegame[n]$ constitutes a strongly connected subgame, and therefore will not be decomposed.
\end{proof}

We compared the performance of the PGSolver tool, a publicly available
tool that contains an implementation of the optimised recursive algorithm,
on the family $\whitegame$ to that of Friedmann's family of games 
(denoted with $\mathcal{F}$), see Figure~\ref{fig:compared}. The figure
plots the number of vertices (horizontal axis) and the time required
to solve the games (vertical log scale axis), clearly illustrating that
$\whitegame$ games are harder.

\begin{figure}[h]
\centering
\begin{tikzpicture}[scale=1.0]
\begin{axis}[ymin=1, ymax=22000, ymode = log, xlabel=Number of vertices, ylabel=runtime (seconds),
legend style={cells={anchor=east}, at={(1.15,0.25)}}]
\addplot[mark=x,only marks] table {White.plot};
\addlegendentry{$\whitegame$ family}
\addplot[mark=o,only marks] table {Friedmann.plot};
\addlegendentry{$\mathcal{F}$ family}

\end{axis}
\end{tikzpicture}
\caption{Runtime of the optimised recursive algorithm (vertical log scale
axis) in seconds versus number of vertices of the games (horizontal axis).}
\label{fig:compared}
\end{figure}
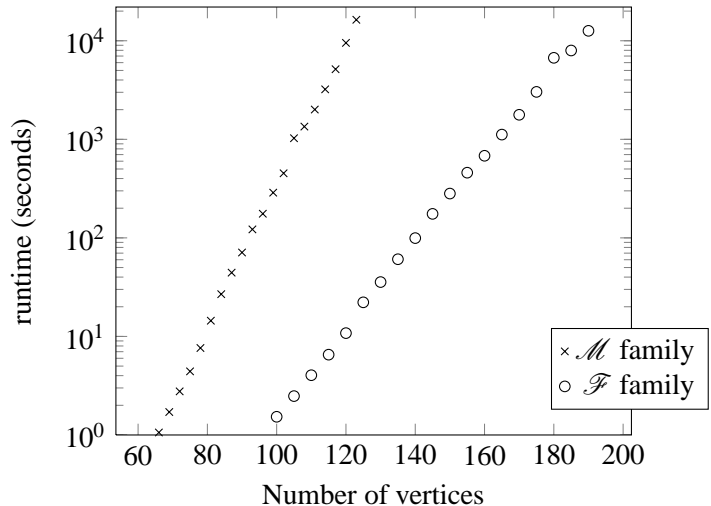

\section{Conclusions} \label{sec:conclusions}

We explored the complexity of solving special parity games using
Zielonka's recursive algorithm, proving that weak games are solved in
polynomial time and dull and nested solitaire games require exponential
time. The family of games $\solitaire$ we used to prove the exponential
lower bounds in addition tighten the lower bound to $\Omega(2^{|V|/3})$
\myomit{, where
$V$ is the set of vertices in the graph,} for the original algorithm
by Zielonka.

We show that a standard optimisation of the algorithm permits solving
all three classes of games in polynomial time. The technique used in
the optimisation (a tight integration of a strongly connected component
decomposition and Zielonka's algorithm) has been previously implemented
in~\cite{FL:09} and was observed to work well in practice. Our results
provide theoretical explanations for these observations.  

We furthermore studied the lower bounds of Zielonka's algorithm
\emph{with} optimisation. In the last section, we improve on Friedmann's
lower bound and arrive at a lower bound of $\Omega(2^{|V|/3})$ for the
optimised algorithm.  For this, we used a family of games $\whitegame$
for which we drew inspiration from the family $\solitaire$ and the games
defined in~\cite{Fri:11}. We believe that an additional advantage of
the families of games $\solitaire$ and $\whitegame$ we defined in this
paper over Friedmann's games lies in their (structural) simplicity.

Our complexity analysis for the special games offers additional insight
into the complexity of Zielonka's algorithm and its optimisation and
may inspire future optimisations of the algorithm. In a similar vein, the same type of analysis can be performed on other parity game solving algorithms from the literature, \eg strategy improvement algorithms.
\bibliographystyle{eptcs}
\bibliography{paper}

\end{document}